\documentclass[letterpaper, 10 pt, conference]{ieeeconf}  

\usepackage{multirow}
\usepackage{array}

\usepackage{graphicx,subcaption}
\usepackage{epstopdf}%
\usepackage{epsfig} 
\usepackage{stackrel}
\usepackage{lipsum}
\usepackage{amssymb}
\usepackage{epstopdf}
\usepackage{epsfig} 
\usepackage[linesnumbered,ruled]{algorithm2e}
\usepackage{calligra}
\usepackage[T1]{fontenc}
\usepackage[noend]{algorithmic}
\usepackage{chngcntr}
\usepackage{mathtools}
\usepackage{mathbbol}
\newcommand{\RN}[1]{%
  \textup{\uppercase\expandafter{\romannumeral#1}}%
}

\newtheorem{theorem}{Theorem}

\newtheorem{lem}{Lemma}

\usepackage{color}
\DeclareMathOperator*{\argmin}{argmin}
\DeclareMathOperator*{\argmax}{argmax}

\newtheorem{prob}{Problem}

\newtheorem{assumption}{Assumption}
\newtheorem{rmk}{Remark}

\begin{document} 

\title{Data-Driven Approximate Abstraction for Black-Box Piecewise Affine Systems}


\author{Gang Chen and Zhaodan Kong}
\maketitle 
\let\thefootnote\relax\footnotetext{Authors are with the Department of Mechanical and Aerospace Engineering, University of California, Davis. Z. Kong is the corresponding author (email: zdkong@ucdavis.edu).}

\begin{abstract}
How to effectively and reliably guarantee the correct functioning of safety-critical cyber-physical systems in uncertain conditions is a challenging problem. This paper presents a data-driven algorithm to derive approximate abstractions for piecewise affine systems with unknown dynamics. It advocates a significant shift from the current paradigm of abstraction, which starts from a model with known dynamics. Given a black-box system with unknown dynamics and a linear temporal logic specification, the proposed algorithm is able to obtain an abstraction of the system with an arbitrarily small error and a bounded probability. The algorithm consists of three components, system identification, system abstraction, and active sampling. The effectiveness of the algorithm is demonstrated by a case study with a soft robot. 
\end{abstract}

\section{Introduction}

The proliferation of cyber-physical systems (CPSs) brings how to effectively and reliably guarantee their correct behaviors to the forefront of problems we as control engineers need to address. One natural choice to attain correct functioning is to consider formal methods techniques, such as model checking \cite{baier2008principles,alur2015principles}, which have been successfully used in the formal verification and synthesis of digital circuits and software codes \cite{peled2013software}. In recent years, we have seen many efforts of extending formal methods to engineering applications, e.g., automobiles \cite{kamali2017formal,abbas2014robustness} and robotics \cite{lopes2016supervisory,ulusoy2014incremental}. One crucial component of formal methods is a precise and potentially concise mathematical model of the system under investigation. However, in reality we rarely have full knowledge of complex CPSs during their design and even testing phase. Thus how to attain formal guarantee for systems with partially or fully unknown dynamics becomes a problem of practical significance. 

In this paper, we aim to address this problem in the context of \textit{abstraction} \cite{baier2008principles}. Given a system model $\mathcal{T}$ (which can potentially have infinitely many states) and a formal specification $\phi$ written, for instance, in linear temporal logic (LTL), an \textit{abstract} model of $\mathcal{T}$ is a simpler model $\mathcal{T}'$; checking whether the simpler model $\mathcal{T}'$ satisfies $\phi$ suffices to decide whether $\mathcal{T}$ satisfies $\phi$ \cite{girard2010approximately,haghverdi2005bisimulation,prabhakar2015stability,pola2015symbolic}. For systems that can be described by discrete state models, abstraction can be achieved by using the concepts of simulation and bi-simulation \cite{baier2008principles,alur2015principles}. In control community, recently there have been many successful efforts pertaining to the abstraction of systems of more realistic dynamics, such as those that are piecewise affine \cite{yordanov2012temporal,gol2014finite,decastro2015dynamics,pola2014symbolic}. All these studies, as far as we know, assume models of known dynamics, which significantly impedes the application of abstraction in the analysis and design of systems with inherent uncertainties, e.g., those needing to interact with a variety of human users and be deployed in a variety of environments.

One principle way of mitigating uncertainties is to utilize machine learning techniques. Actually, the integration of formal methods and machine learning has shown great potential in the formal specification, design, verification, and validation of CPS \cite{dreossi2017compositional,urayama2015simulation,seshia2016design,haesaert2015data,kozarev2016case,kong2017temporal}. In this paper, we will focus on how to combine machine learning techniques, particularly system identification and active learning, and formal methods techniques to generate approximate abstractions for systems with black-box (unknown), piecewise affine (PWA) dynamics. PWA models partition the state space into a finite number of polyhedral regions and consider affine dynamics in each region \cite{bemporad2005bounded}. It has been shown that PWA models can approximate nonlinear dynamics with arbitrary accuracy \cite{yordanov2012temporal}. Moreover, there exist efficient techniques for the identification of PWA systems, e.g., optimization-based methods and clustering-based methods \cite{garulli2012survey}. 

The major contribution of this paper is that it addresses many theoretical and algorithmic issues pertaining to the integration of existing approximate abstraction techniques \cite{yordanov2013formal} and system identification techniques \cite{garulli2012survey}. Given a system with unknown PWA dynamics, the paper shows that, by following the algorithm prescribed in the paper, it is possible to extract an abstract model with an arbitrarily small error and a bounded probability (under certain mild assumptions). Even though the paper focuses on PWA systems, the preliminary results obtained in it can potentially pave the way for future developments for systems with more complex dynamics. 

The remainder of the paper is organized as follows. In Section \ref{sec2}, we provide preliminaries and notation used throughout the paper. In Section \ref{sec3}, we formally introduce the abstraction problem that will be solved in the paper. Section \ref{sec4} presents our data-driven approximate abstraction algorithm, together with proofs demonstrating the effectiveness of our algorithm. Section \ref{sec5} uses a soft robot system as an example to showcase our proposed algorithm. We conclude with final remarks in Section \ref{sec6}.
    
\section{Preliminaries and Notation}
\label{sec2}

A $N$ dimensional \textit{polytope} $\mathcal{X}$ is defined as the convex hull of at least $N+1$ affinely independent vectors in $\mathbb{R}^{N}$. A complete partition of $\mathcal{X}$ is a set of open polytopes $\mathcal{X}_{i}, i\in I$ ($I$ is a finite index set) in $\mathbb{R}^{N}$ such that $\mathcal{X}_{i_1}\cap \mathcal{X}_{i_2}=\emptyset$ for all $i_{1}, i_{2} \in I , i_{1}\neq i_{2}$ and $cl(\mathcal{X})=\cup_{i\in I} cl (\mathcal{X}_{i})$, where $cl(\mathcal{X}_{i})$ denotes the closure set of $\mathcal{X}_{i}$. According to the $H$-representation, each $\mathcal{X}_{i}, i\in I$ can be represented as $\mathcal{X}_{i}= \{ x \in \mathbb{R}^N: H_{i}x \prec K_{i} \}$, where $\prec$ denotes componentwise inequality. 

A \textit{piecewise affine (PWA) system} \cite{bemporad2005bounded} can be written as follows:
\begin{equation}
\label{PWAARX}
\begin{array}{rl}
x_{k+1}&=f(x_{k})+e\\
f(x)&=\left\lbrace 
\begin{array}{lll}
A_{1}x+b_{1} \text{ if  }  x\in \mathcal{X}_{1}\\
\vdots\\
A_{s}x+b_{s} \text{ if  }  x\in \mathcal{X}_{s}\\
\end{array}\right.
\end{array}
\end{equation}
where $x_k$  is the state of the system at step $k$; $f:\mathcal{X} \rightarrow \mathcal{R}^{N}$ is a PWA map; $e \in \mathcal{N}(0,\sigma_{e}^2)$ is an independently, identically distributed and zero mean Gaussian noise with standard deviation $\sigma_{e}$; $s$ is the number of modes; $A_{i}, b_{i}$ are the parameters of the $i$-th mode ($A_{i}, i=1,\cdots,s$ is assumed to be nonsingular in this paper); and all $s$ modes together constitute a complete partition of $\mathcal{X}$. 

A \textit{transition system} is a tuple $\mathcal{T}=(Q, \delta, O, o)$, where $Q$ is the state space; $\delta :Q \rightarrow 2^{Q}$ ($2^{Q}$ is the powerset of $Q$) is a transition map assigning a state $q \in Q$ to its next state $q'\in Q$; $O$ is the set of observations; and $o :Q\rightarrow O$ is an observation map assigning each $q\in Q$ an observation $o(q)\in O$ \cite{baier2008principles}. We denote a region of the state space as $P\subset Q$. The \textit{embedding transition system} of a PWA system  $\mathcal{S}$ described by Eqn. (\ref{PWAARX}) is a tuple $\mathcal{T}_{e}=(Q_{e}, \delta_{e}, O_{e},o_{e})$, where $Q_{e}=\cup_{i\in I}\mathcal{X}_{i}$; $\delta_{e}: x\rightarrow x'$ iff there exists $i\in I$ such that the transition from $x$ to $x'$ satisfies Eqn. (\ref{PWAARX}); $O_{e}=I$; and $o_{e}(x)=i \text{ iff } x\in \mathcal{X}_{i}$ \cite{yordanov2013formal}. 

For a transition system $\mathcal{T}$, including embedding transition systems of PWA systems, the \textit{successor} of a region $P \subset Q$ is define as the set of states that can be reached from the states in $P$ in one step, i.e., $Post(P)=\{ q\in Q \mid \exists p\in P \text{ with } p\rightarrow  q\}$. The \textit{predecessor} of a region $P \subset Q$ can be defined similarly as $Pre(P)=\{ q\in Q \mid \exists q\in P \text{ with } q\rightarrow  p\}$. A state $q \in Q$ is called \textit{reachable} if there exists a finite execution ending at $q$. We denote all the reachable states of $\mathcal{T}$ by $Reach(\mathcal{T})$. Given an LTL formula $\phi$ over $O$ and a system $\mathcal{T}$, if all the traces originating from a region $P \subset Q$ satisfy $\phi$, then we denote the situation as $\mathcal{T}(P)\models \phi$. Let $X^{\phi}_{\mathcal{T}}=\{X \subset Q, \mathcal{T}(X)\models \phi\}$ denote the largest region of $\mathcal{T}$ from which all the traces satisfy $\phi$ \cite{baier2008principles}. 

The \textit{reachability metric} over two transition systems $\mathcal{T}_{1}$ and $\mathcal{T}_{2}$ is defined as \cite{girard2007approximation}: 
\begin{equation*}
d(\mathcal{T}_{1},\mathcal{T}_{2})=h(Reach(\mathcal{T}_{1}),Reach(\mathcal{T}_{2})),
\end{equation*}
where $h$ is the Hausdorff distance. Given two transition systems $\mathcal{T}_{1}$ and $\mathcal{T}_{2}$ with the same observation set $O$ and a reachability metric $d$ defined over them, a relation $\mathcal{S}_{\sigma}\subseteq Q_{1}\times Q_{2}$ is called a $\sigma-$\textit{approximate simulation relation} of $\mathcal{T}_{1}$ by $\mathcal{T}_{2}$ \cite{girard2007approximation} if for all $(q_{1},q_{2})\in \mathcal{S}_{\sigma}$:
\begin{itemize}
\item $d(o(q_{1}),o(q_{2}))\leq\sigma$,
\item $\forall q_{1}'=Post(q_{1})$, there exists $q_{2}'=Post(q_{2})$, such that $ (q_{1}',q_{2}')\in \mathcal{S}_{\sigma}$.
\end{itemize}
Moreover, $\mathcal{T}_{1}$ is said to be $\sigma-$\textit{approximately simulated} by $\mathcal{T}_{2}$, denoted $\mathcal{T}_{1} \prec_{\sigma} \mathcal{T}_{2}$.


\section{Problem Statement}
\label{sec3}

Formally, in this paper, we wish to solve the following problem:
\begin{prob}
Given a PWA system $\mathcal{S}$ with unknown dynamics, an LTL specification $\phi$, and a bound $\sigma>0$, find a finite transition system $\hat{\mathcal{T}}$ such that $p(\hat{\mathcal{T}} \prec_{\sigma} \mathcal{T})>1-\delta$, where $p(.)$ stands for probability, $\mathcal{T}$ is the true abstract transition system of $\mathcal{S}$, and $\delta$ is bounded.
\label{mainproblem}
\end{prob}

\begin{rmk}
By unknown dynamics, we mean that the following system characteristics are unknown: (i) the number of modes, $s$, (ii) the parameters related to the dynamics of each mode, $\{A_i, b_i, i=1,\cdots,s\}$, (iii) the parameters related to the partitions (regions) of the state space, $\{H_i, K_i, i=1,\cdots,s\}$, and (iv) the standard deviation of the Gaussian noise, $\sigma_e$. But we assume that our algorithm, which will be presented in the next section, can use the system as a black-box simulator to generate samples. This is a reasonable assumption since during system design and testing phases, engineers can always get access to a full-scale system model, a scaled system model, or a computer simulation to generate samples \cite{jin2015mining,kong2017temporal}.
\end{rmk}

\begin{rmk}
Notice that the requirement $p(\hat{\mathcal{T}} \prec_{\sigma} \mathcal{T})>1-\delta$ is inspired by the concept of probably approximately correct (PAC) models in machine learning \cite{valiant2013probably}. It simply says that the probability that the transition system $\hat{\mathcal{T}}$ (obtained by using our algorithm) is a $\sigma$-approximate simulation by the PWA system $\mathcal{S}$ is higher than $1-\delta$. In other words, given a PWA system with unknown dynamics, we intend to find out its approximate abstract transition system $\hat{\mathcal{T}}$ with a high enough confidence.  
\end{rmk}

\section{Data-Driven Abstraction Algorithm}
\label{sec4}

\begin{figure}[hbtp!]
\centering
\includegraphics[width=0.5\textwidth]{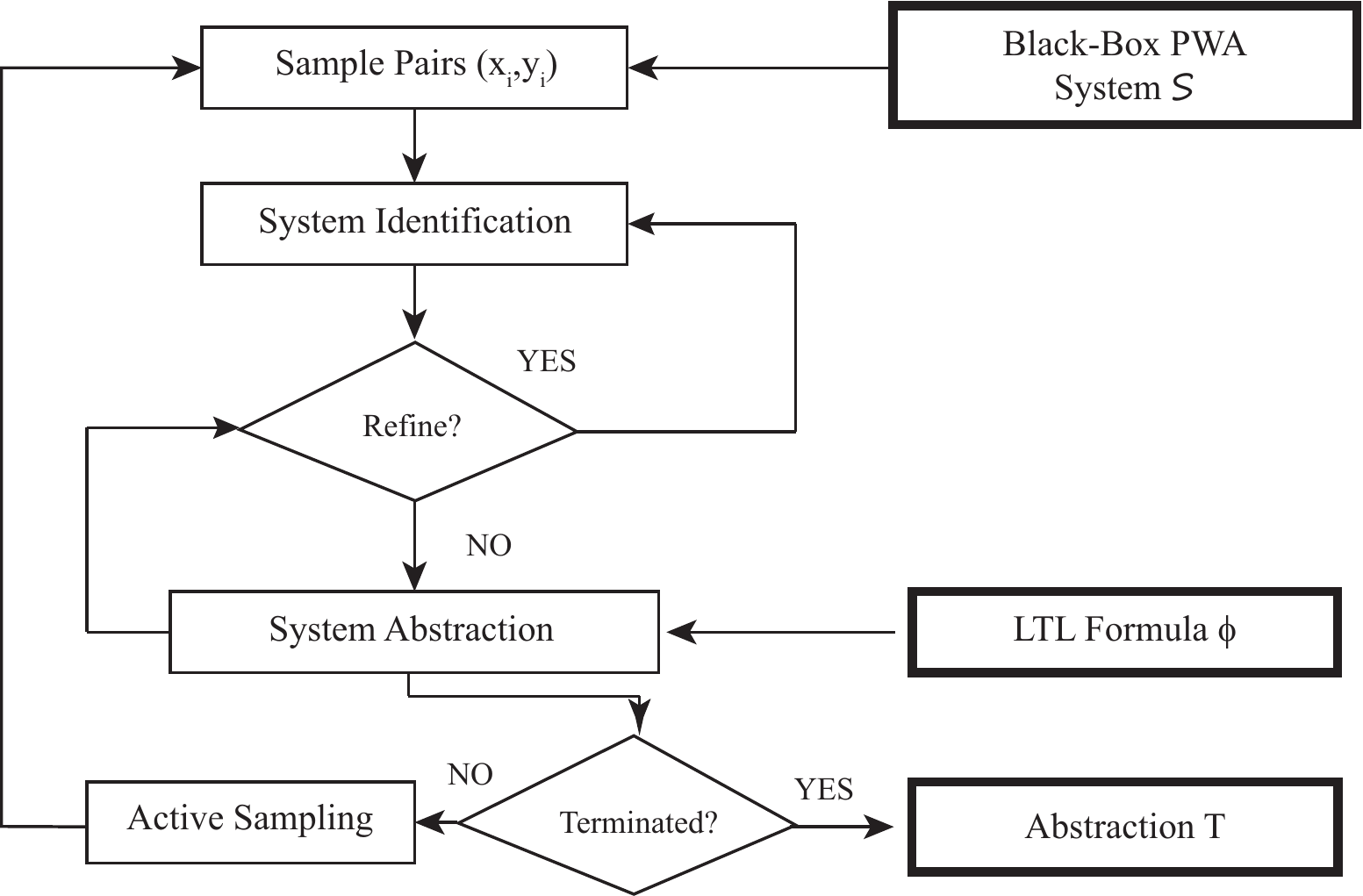}
\caption{Architecture of our data-driven abstraction algorithm.}
\label{overall}
\end{figure}

Fig. \ref{overall} illustrates the basic architecture of our data-driven abstraction algorithm to solve Problem \ref{mainproblem}. The inputs of the algorithm are a black-box PWA system $\mathcal{S}$ with unknown dynamics and an LTL specification $\phi$; the output of the system is a transition system $\mathcal{T}$. The algorithm can be roughly divided into three components: system identification, system abstraction, and active sampling. The goal of the \textit{system identification} component is to derive an estimated PWA model $\hat{\mathcal{S}}$ based on the data sampled from the black-box system $\mathcal{S}$ (serving as a simulator); the goal of the \textit{system abstraction} component is to derive a transition system $\mathcal{T}$ given the identified PWA model $\hat{\mathcal{S}}$ and the specification $\phi$; one important procedure of the system identification component is refinement, which refines the estimated model $\hat{\mathcal{S}}$, until no significant improvement can be achieved, based on the currently available data and the current abstraction $\mathcal{T}$; finally, if no satisfactory abstraction $\mathcal{T}$ can be found after the refinement, the \textit{active sampling} component will be implemented to draw new data points with the help of the black-box simulator $\mathcal{S}$. 

In the following sub-sections, we will present each of the three components. Proofs regarding the effectiveness of our algorithm will be provided at the end of the section.   

\subsection{System Identification}
\label{sub:systemid}

Given a black-box PWA system $\mathcal{S}$, or subsequently a set of $K$ samples $\mathcal{D}:=\{y_k, x_k\}, k =1,\cdots, K$, the system identification component identifies a PWA model, specified by the number of modes $s$ as well as the mode parameters $\hat{A}_{i},\hat{b}_i, \hat{H}_{i}$ and $\hat{K}_{i}$ with $i=1,\cdots, s$. The problem itself is a well-studied problem. Specifically, we need to find (i) a minimum positive integer, $s$, (ii) a set of parameter matrices, $\{\hat{A}_{i}\}_{i=1}^{s}$ and $\{\hat{H}_{i}\}_{i=1}^{s}$, and (iii) a set of parameter vector $\{\hat{b}_{i}\}_{i=1}^{s}$ and $\{\hat{K}_{i}\}_{i=1}^{s}$ (notice that $\{\hat{H}_{i}\}_{i=1}^{s}$ and $\{\hat{b}_{i}\}_{i=1}^{s}$ together constitute a complete partition $\{\mathcal{X}_{i}\}_{i}^{s}$ of the PWA system's state space $\mathcal{S}$), such that the estimated parameters are the solution of the following minimization problem:
\begin{equation}
(\hat{A}_{i}, \hat{b}_{i}, \hat{H}_{i}, \hat{K}_{i})=\argmin_{(x_{k},y_{k})\in (\mathcal{D}\cap\mathcal{X}_{i})}\frac{1}{K}\sum_{k=1}^{K}c(y_{k}-\hat{f}(x_{k}))
\label{costfunction:1}
\end{equation}
where $\hat{f}(.)$ is specified by $\hat{A}_{i}, \hat{b}_{i}, \hat{H}_{i}$ and $\hat{K}_{i}, i=1,\cdots,s$ and $c$ is a given penalty function, which is chosen to be $c(\cdot)=||\cdot||_{2}$ in this paper. Notice that solving the identification problem involves the simultaneous solving of two sub-problems, data classification and parameter estimation. Once the data points have been classified into clusters $\{\mathcal{D}_{i}\}_{i=1}^{s}$ such that $(y_{k},x_{k})\in \mathcal{D}_{i}$, i.e., $(y_{k},x_{k})$ is attributed to the $i$-th mode, mode parameters $\hat{A}_{i},\hat{b}_i, \hat{H}_{i}$ and $\hat{K}_{i}$ can be easily estimated by solving Eqn. (\ref{costfunction:1}). 

Our system identification component is modified from the method proposed in \cite{bemporad2005bounded}. It consists of two main procedures: initialization and refinement. One major difference between our method and the one in \cite{bemporad2005bounded} is that we utilize the current abstract transition system to guide the refinement.

\begin{algorithm} 
    \SetKwInOut{Input}{Input}
    \SetKwInOut{Output}{Output}
 \caption{System Identification Initialization}
 \Input{A bound $\hat{\sigma}$, a set of samples $\mathcal{D}:=\{(y_k,x_k)\}, k \in \{1,\cdots, K\}$, a ratio $0<r<1$, and a large number $J$}
 \Output{$\hat{A}_{i}, \hat{b}_{i},\mathcal{D}_{i}, \hat{\mathcal{H}}_{i}, \hat{\mathcal{K}}_{i}, i=1,\cdots,s$}
\begin{algorithmic}[1]
\STATE Set $l=0$ and $I_0 = \{1,\cdots, K\}$;
  \FOR{$|I_{l+1}|\geq r K$}
  \STATE Randomly generate a set of parameters $\{A_{j},b_{j}\}, j\in \{1,\cdots,J\}$; 
  \STATE Construct sets $\Sigma_{j}=\{\parallel y_{k}-(A_j x_{k} + b_j) \parallel\leq \hat{\sigma}, k\in I_{l}\}$ for each $j\in \{1,\cdots,J\}$;
  \STATE Set $l=l+1$ and $\Sigma_{max}= \argmax_{j=1}^J |\Sigma_{j}|$;
  \STATE With $\Sigma_{max}$, estimate $\hat{A}_{l}$ and $\hat{b}_{l}$ by solving
\begin{equation*}
(\hat{A}_{l}, \hat{b}_{l}) = \argmin_{A, b} \sum_{k=1}^{|\Sigma_{max}|} c(y_{k}-(A x_{k} + b));
\end{equation*}  
  \STATE Set $\mathcal{D}_{l}=\{k\in I_{l}:\parallel y_{k}-(\hat{A}_{l} x_{k}+\hat{b}_{l})\parallel \leq \hat{\sigma}\}$ and $I_{l+1}=I_{l}\setminus \mathcal{D}_{l}$;
 \ENDFOR 
  \STATE Find boundaries specified by $\{\hat{H}\}_{i=1}^l$ and $\{\hat{K}\}_{i=1}^l$ between sets $\{\mathcal{D}_{i}\}_{i=1}^{l}$; set $s=l$.
\end{algorithmic}
\label{algorithm:systemidentification}
\end{algorithm}

\subsubsection{Initialization}

The pseudo code of the initialization procedure is shown in Alg. \ref{algorithm:systemidentification}. The steps are rather self-explanatory. Here we would like to provide a few simple comments for clarification. We need to randomly generate a matrix $\hat{A}_{j}$ and a vector $\hat{b}_{j}$ (Line 3) in each loop, which is quite inefficient; thus the termination condition $|I_{l+1}|\geq rK$ (Line 2) can be set loosely, i.e., with a rather large $r$. Such a practice is reasonable, given the fact that the initialization procedure is only meant to generate some good enough partitions (modes), which will be further refined in the refinement procedure. As for the boundaries specified by $\{\hat{H}\}_{i=1}^l$ and $\{\hat{K}\}_{i=1}^l$ between sets $\{\mathcal{D}_{i}\}_{i=1}^{l}$ (Line 8), standard support vector machines (SVM) regression methods \cite{cortes1995support} can be used to compute them. 

\begin{algorithm} 
\SetKwInOut{Input}{Input}
\SetKwInOut{Output}{Output}
 \Input{Current abstract transition system $\mathcal{T}$, parameters $\hat{A}_{i},\hat{b}_{i}, \mathcal{D}_{i}, \hat{H}_{i},\hat{K}_{i}$, $i=1,\cdots,s$ obtained in Alg. \ref{algorithm:systemidentification}, and a bound $\hat{\sigma}$}
 \Output{$\hat{A}_{i},\hat{b}_{i},\hat{H}_{i},\hat{K}_{i},i=1,\cdots,s$}
  \begin{algorithmic}
   \STATE{Set $\beta \geq 0 ,\mu\geq 0,\kappa \geq 0,0< \theta <1$, $l=1$;}\\
 \end{algorithmic} 
 \While{Not terminated}{
 \begin{algorithmic}[1]
    \STATE{Compute $(i^{*},j^{*})=\argmin_{1\leq i<j\leq s}\beta_{i,j}$ with\\
           $\beta_{i,j}=\parallel \hat{A}_{i}-\hat{A}_{j}\parallel$;}
    \IF{$\beta_{i^{*},j^{*}}\leq \theta^{t}\beta$}
    \STATE{Merge modes $i^{*}$ and $j^{*}$;  Set $s=s-1;$}
    \STATE{Recompute $\mathcal{D}_{i^{*}}=\{k \in \{1,\cdots, K\}: $ $\parallel y_{k}-(\hat{A}_{i^{*}} x_{k}+\hat{b}_{i^{*}})\parallel \leq \hat{\sigma}\}$; \label{compute}}
    \ENDIF
    \STATE{Use $CR()$ and rule $Dis( )$ to reassign points;\\
    \textit{$/*CR()$ and $Dis( )$ are defined in the text as\\ Eqn. (\ref{eqn:CR}) and Eqn. (\ref{eqn:Dis}), respectively$*/$}} 
    \STATE{Compute $i^{*}=\argmin_{i=1,\cdots,s} |\mathcal{D}_{i}|/|\mathcal{D}|$;}
    \IF{$ |\mathcal{D}_{i}|/|\mathcal{D}| \leq \theta^{t}\mu$}
    \STATE{Discard mode $i^{*}$; let $s=s-1$;}
    \STATE{Go to Step \ref{compute};}
    \ENDIF
    \STATE{Store $\{\hat{A}_i\}_{i=1}^s$ as $\{\hat{A}_i^{old}\}_{i=1}^s$;}
    \STATE{Update $\hat{A}_{i},\hat{b}_{i}, \hat{H}_{i},\hat{K}_{i}$ with new $\{\mathcal{D}_{i}\}_{i=1}^s$;}
    \IF{$\parallel \hat{A}_{i}- \hat{A}_{i}^{old}\parallel \leq \kappa$}
    \STATE{Terminated.}
    \ELSE
    \STATE{$l=l+1$;}
    \ENDIF
    \end{algorithmic}
 }
  \caption{System Identification Refinement}
  \label{algorithm:refinesystemid}
\end{algorithm}

\subsubsection{Refinement}

Given the random nature of the way $\hat{A}_{i},\hat{b}_{i}, i=1,\cdots,s$ are generated in Alg. \ref{algorithm:systemidentification}, it is quite unlikely that we are able to identify all the correct modes with the initialization procedure. Potentially there are two main issues:
\begin{itemize}
\item \textit{Undecidable data points}: these are the data points that belong to more than one mode, i.e., they satisfy $\parallel y_{k}-(\hat{A}_i x_{k} + \hat{b}_i)\parallel \leq \hat{\sigma} $ for more that one $i=1,\cdots,s$;
\item \textit{Unfeasible data points}: these are the data points that don't belong to any mode, i.e., there is no $i=1,\cdots,s$ such that $\parallel y_{k}-(\hat{A}_i x_{k} + \hat{b}_i)\parallel \leq \hat{\sigma}$. 
\end{itemize}

We refine the identified system model by eliminating these two types of data points as follows:

For \textit{undecidable points}, we can reassign them based on their maximum likelihood with respect to all modes, i.e., we assign each undecidable point $(x_{k},y_{k})$ to the optimal mode according to the following rule:
\begin{equation}
\label{eqn:CR}
CR(k)=\argmax_{i=1,\cdots,s} CC(x_{k},y_{k})_{i}/CC(x_{k},y_{k}),
\end{equation}
where $CC(x_{k},y_{k})$ is the total number of data points that are within $\mathcal{D}$ and inside the hyper ball region centered at $(x_{k},y_{k})$ and with a predefined radius $\rho$, and $CC(x_{k},y_{k})_{i}$ is the number of data points that are in $CC(x_{k},y_{k})$ and belong to mode $i$. 

For \textit{unfeasible points}, given the current abstract transition system $\mathcal{T}$, we discard those points that meet the following condition:
\begin{equation}
\label{eqn:Dis}
\begin{array}{rl}
&Dis(k)\\
&=\{(x_{k},y_{k})\in \mathcal{D}| d(y_{k}, o^{-1}(Post(o(x_k))))\geq \hat{\sigma}\},
\end{array}
\end{equation}
where $ d(., .)$ is the Hausdorff distance, $o(x_k)$ maps a continuous state $x_k$ to a discrete state of the transition system $\mathcal{T}$(it would be helpful for the readers to review the definitions in Section \ref{sec2}), $Post(o(x_k))$ maps the region corresponding to $o(x_k)$ to its successor region, and finally $o^{-1}(Post(o(x_k))))$ finds the set of continuous, PWA states corresponding to the region $Post(o(x_k))$. The remaining unfeasible data points are reassigned according to the rule $CR(\cdot)$ (Eqn. (\ref{eqn:CR})). The pseudo code of the whole refinement procedure is shown in Alg. \ref{algorithm:refinesystemid}.

\subsection{System Abstraction}
\label{abstraction}

Given an estimated PWA model $\hat{\mathcal{S}}$ (or $\hat{f}$), parameterized by $\hat{A}_{i},\hat{b}_i, \hat{H}_{i}$ and $\hat{K}_{i}$ with $i=1,\cdots, s$, and an LTL formula $\phi$, the goal of the system abstraction component is to generate a good enough abstraction $\mathcal{T}$. We roughly follow the approximate abstraction procedures described in \cite{yordanov2013formal} to design and implement the system abstraction component. Here we are just going to provide a rough outline of the abstraction algorithm. Interested readers can refer to \cite{yordanov2013formal} for more details. First, a deterministic Buchi automaton $\mathcal{B}_{\phi}$ is constructed from the formula $\phi$. Second, the corresponding embedding transition system $\mathcal{T}_e$ is constructed for $\hat{\mathcal{S}}$ by simply using the definition of embedding transition system. Third, an observation map $o_{e}$ is created by partitioning the state space of the system $\mathcal{T}_e$ into uniform grids. Fourth, given the system $\mathcal{T}_e$, the observation map $o_{e}$, and the LTL formula $\phi$ (or its corresponding Buchi automaton $\mathcal{B}_{\phi}$), an initial transition system $\mathcal{T}_0$ is constructed by following standard abstraction procedures, such as those prescribed in \cite{baier2008principles}. Fifth, a product automaton $\mathcal{P}$ is constructed as $\mathcal{P}=\mathcal{T}_0 \times \mathcal{B}_{\phi}$, which concerns both the initial transition system $\mathcal{T}_0$ and the specification $\phi$. The product automaton is a tuple $\mathcal{P}=(S_{p},S_{p0},\delta_{p},F_{p})$, where $S_{p}$ is the set of states, $S_{p0}$ is the set of initial states, $\delta_{p}$ is the transition map, and $F_{p}$ is the acceptance condition. Finally, refinement is conducted by solving a deterministic Rabin game. 

\begin{algorithm} 
\SetKwInOut{Input}{Input}
\SetKwInOut{Output}{Output}
 \Input{Current abstract transition system $\mathcal{T}$, current product automaton $\mathcal{P}$, an initial state $q$ of $\mathcal{T}$, and a ratio $0<\eta<1$}
 \Output{Refined abstract transition system $\hat{\mathcal{T}}$ and refined product automaton $\hat{\mathcal{P}}$}
  \begin{algorithmic}
 \STATE{Initialize $\hat{\mathcal{T}}=\mathcal{T},\hat{\mathcal{P}}=\mathcal{P}$, and $S_{u} = \emptyset$;}\\
 \end{algorithmic} 
 \While{$|S_{u}| \geq \eta|Q|$}{
 \begin{algorithmic}[1]
\STATE{Compute  $S_{\top}$ and $S_{\bot}$ for $\hat{\mathcal{P}}$;} 
\STATE{Set $S_{u}:=\mathcal{P}_{p}\setminus (S_{\top}\cup S_{\bot})$;\\
\textit{$/*$See text for the definitions of $S_{\top}$, $S_{\bot}$,\\ and $S_{u}*/$}}
\FORALL{$(q,g)\in S_{u} $}
\IF{$q \in S_u$}
\STATE{Set $\mathbb{q}:={q}$;}\\
\FORALL{($\exists q_{r}\in \mathbb{q},q'\in Post(q_{r})$ and $(q_{r}\cap Pre(q'))\neq \emptyset$)}
\STATE {Construct states $q_{1},q_{2}$ such that\\
$q_{1}:=q_r\cap Pre(q')$,\\
$q_{2}:=q_{r}\setminus Pre(q')$;}
\STATE{$\mathbb{q}:=(\mathbb{q}\setminus q_{r})\cup \{q_{1},q_{2}\}$;}
\STATE{$q:=\mathbb{q}$;}
\STATE{Update $\delta$ and $o$ for $\hat{\mathcal{T}}$;}
\ENDFOR
\ENDIF
\ENDFOR
\STATE{Update $\hat{\mathcal{P}}$ and $\hat{\mathcal{T}}$.}
 \end{algorithmic}
 }
  \caption{System Abstraction Refinement}
  \label{algorithm:refineabstraction}
\end{algorithm}

The pseudo code of the abstraction refinement procedure is shown in Alg. \ref{algorithm:refineabstraction}. Given the current abstract transition system $\mathcal{T}$ and the current product automaton $\mathcal{P}$, a state $q \in Q$ of $\mathcal{T}$ falls into one of the following three categories: 
\begin{itemize}
\item $S_{\top}$: the set of states from which all traces are accepted by $\mathcal{P}$, 
\item $S_{\bot}$: the set of states from which no trace is accepted by $\mathcal{P}$,
\item $S_{u}$: the set of states from which some but not all traces are accepted by $\mathcal{P}$. 
\end{itemize}
The goal of the refinement is thus to eliminate $S_{u}$. This leads to the termination condition of the refinement procedure as $|S_{u}| < r |Q|$, i.e., the refinement will be terminated once the volume of $S_{u}$, $|S_{u}|$, is smaller than a fraction $r$ (specified by the user) of the volume of $Q$ (the state space of $\mathcal{T}$).   

\subsection{Active Sampling}

The system identification component and the system abstraction component described in the last two sub-sections are based on a fixed data set $\mathcal{D}:=\{(y_{k},x_{k})\}_{k=1}^{K}$. It is quite obvious that the quality of the system identification and the system abstraction depends on the quality of the data set. To improve the quality of the system identification (in other words, to decrease the number of data points needed for the system identification), we use an active learning algorithm developed by our group \cite{chen2016active} to sample high quality (or ``informative'') data points for the system identification component after the initial round (see Fig. \ref{overall}). 

The strategy to find the next data point to sample for round $t+1$ has two steps. In the first step, the best candidate for each mode is identified as follows:
\begin{equation}
x_{i}:=\argmax_{x \in \hat{\mathcal{X}}_{i}}( \Psi_{i,t}(x)+\lambda^{1/2}_{t} \varrho_{i,t}(x))
\label{activelearning}
\end{equation}
where $\Psi_{i,t}(x)$ is the Gaussian process regression mean of the prediction error defined over the data points in $\mathcal{D}_{i}$, $\varrho_{i,t}(x)$ is the Gaussian process regression variance of the prediction error defined over the data points in $\mathcal{D}_{i}$, and $\lambda_{t}$ is a regularization factor. In the second step, the active learning algorithm chooses mode $i^{*}=\argmin_{i=1,\cdots,s}(\max_{x\in\hat{\mathcal{X}}_{i}}\Psi_{i,t}(x))$ and the corresponding best candidate $x_{i^*}$ to sample.

\subsection{Theoretical Results Regarding the Effectiveness of Our Algorithm}

We add the following assumption regarding the performance of the PWA system identification.

\begin{assumption}
\label{assumption1}
Assume the prediction error of the system identification component described in Section \ref{sub:systemid} can be characterized by a zero mean Gaussian with a bounded variance, i.e., $f(x)-\hat{f}(x) \sim \mathcal{N}(0,\sigma_{p}(x)^{2})$ and $\sigma_{p}(x)\leq C$, where $f(x)$ is the real PWA dynamics and $\hat{f}(x)$ is the estimated PWA dynamics.
\end{assumption}

\begin{rmk}
The system identification of PWA system is still an open problem and has been proven to be NP-hard \cite{ljung2010perspectives,lauer2015complexity}. In \cite{bemporad2005bounded}, the authors were able to demonstrate that, for a fixed data set, the error is bounded, i.e., $|f(x)-\hat{f}(x)|<\sigma$ for any $\sigma>0$. Thus, we believe our assumption here, even though unproven, is still reasonable.
\end{rmk}

We have the following three lemmas regarding the integration of identification and abstraction (without active sampling in the loop). 

\begin{lem}
Given a PWA system $\mathcal{S}$ with known dynamics, for any bound $\varepsilon>0$, Alg.\ref{algorithm:refineabstraction} can derive an abstract transition system $\hat{\mathcal{T}}$ that is $\varepsilon-$approximately simulated by the real abstract transition system $\mathcal{T}$ of $\mathcal{S}$, i.e., $\hat{\mathcal{T}} \prec_{\varepsilon} \mathcal{T}$.
\label{abstractionerror}
\end{lem}

\begin{proof}
 Consider a relationship $V(q_{1},q_{2})=\{(q_{1},q_{2})|d(q_{1},q_{2})\leq\varepsilon\}$, and set $ r|Q|\leq \varepsilon$. As $(q_{1},q_{2})\in V$ implies $d(q_{1},q_{2})\leq \varepsilon$, the first condition in the definition of approximate relation is satisfied.  Then for all $d(q_{1},q_{2})<\varepsilon$, the conclusion in \cite{yordanov2013formal} shows that the result of Alg.3 can guarantee that $d(Post(q_{1}),Post(q_{2}))\leq|S_{u}|\leq\varepsilon$. Then the second condition in the definition of approximate relation is satisfied. Therefore, $\hat{\mathcal{T}} \prec_{\varepsilon}\mathcal{T}$.
\end{proof}

\begin{lem}
Given an estimation $\hat{f}(.)$ (or $\hat{\mathcal{S}}$) of the PWA system $f(.)$ (or $\mathcal{S}$), if $\parallel \hat{f}(x)-y\parallel\leq \eta_{1}$ holds with probability $> 1-\alpha_{1}$, where $\alpha_{1}\in (0,1), \eta_{1}>0$, then the system abstraction component described in Section \ref{abstraction} can derive an abstract transition system $\hat{\mathcal{T}}$ such that it is $(\eta_{1}+\varepsilon)-$approximately simulated by the real abstract transition system with a probability greater than $1-\alpha_{1}$.
\label{Lemma abstraction bound}
\end{lem}

\begin{proof}
Here for the sake of clarify, the notations in this proof that are slightly different from the ones used in the other parts of the paper. Let's use $f$ and $\hat{f}$ to denote the real and estimated dynamics of the PWA system, respectively. Moreover, let's call their true abstract transition systems as $\mathcal{T}_{real}$ and $\hat{\mathcal{T}}_{real}$, respectively. Finally, let's call the approximate abstract transition systems obtained by using the system abstraction component described in Section \ref{abstraction} as $\mathcal{T}_{app}$ and $\hat{\mathcal{T}}_{app}$, respectively. According to Lemma 1, we have $\mathcal{T}_{app} \prec_{\varepsilon} \mathcal{T}_{real}$ and $\hat{\mathcal{T}}_{app} \prec_{\varepsilon} \hat{\mathcal{T}}_{true}$. Then set  $\parallel e_{B} \parallel=\parallel \hat{f}(x)-y\parallel\leq \eta_{1}$, such that $y=\hat{f}(x)+y-\hat{f}(x)= \hat{f}(x)+e_{B}$, which can be seen a PWA system with bounded noise,  following \cite{yordanov2013formal}, we have $\hat{\mathcal{T}}_{real} \prec_{\eta_{1}} \mathcal{T}_{real}$. Since the relationship $\prec_{.}$ is transitive, we have $\hat{\mathcal{T}}_{app} \prec_{\eta_{1}+\varepsilon} \mathcal{T}_{real}$. The conclusion regarding probability follows easily.
\end{proof}

\begin{lem}
Provided with Assumption \ref{assumption1}, $\parallel \hat{f}(x)-y\parallel\leq \eta_{2}$ holds with probability $\geq 1-\alpha_{2}$, where $\alpha_{2}=1-\frac{1}{\sqrt{2\pi(\sigma_{e}+C)}}\int^{\eta_{2}}_{-\eta_{2}}exp(-\upsilon^{2}/(2(\sigma_{e}+C)))d\upsilon$.
\label{Lemma bounded error probability}
\end{lem}

\begin{proof}
Set $e_{T}:=y-\hat{f}(x)=f(x)+e-\hat{f}(x)$, where $e$ is the true Gaussian noise, $e\sim \mathcal{N}(0,\sigma_{e}^2)$. Since $f(x)-\hat{f}(x)\sim \mathcal{N}(0,\sigma_{p}(x)^2)$, we have $e_{T}\sim \mathcal{N}(0,\sigma_{p}^2+\sigma_{e}^2)$. The probability that $\parallel \hat{f}(x)-y\parallel\leq \eta_{2}$ can then be computed. When $\sigma_{p}=C$, the probability reach the minimum value, such that $\alpha_{2}=1-\frac{1}{\sqrt{2\pi(\sigma_{e}+C)}}\int^{\eta_{2}}_{-\eta_{2}}exp(-\upsilon^{2}/(2(\sigma_{e}+C)))d\upsilon$.
\end{proof}

We have the following two lemmas regarding the performance of our active sampling component. 


\begin{lem}
For any $\rho\in (0,1)$, if $\lambda_{t}=2B+300\gamma_{t}\text{log}^{3}(t/\rho)$, where $\gamma_{t}$ is the maximum information gain defined in \cite{chen2016active}, then 
\begin{equation}
\parallel \Psi_{i,t}(x) \parallel \leq \lambda_{t}^{1/2}\varrho_{i,t}(x)
\end{equation}
holds with probability $> 1-\rho$.
\label{activeerrorbound}
\end{lem}

\begin{proof}
The lemma can be proved by following similar steps as the proof of Lemma 2 in \cite{chen2016active}.
\end{proof}

\begin{lem}
For any $\sigma > 0$, if $\lambda_{t}=2B+300\gamma_{t}\text{log}^{3}(t/\rho)$, where $t$ is the number of points in $\mathcal{D}_{i}$, then there exists $t \leq T$, 
\begin{equation}
 \lambda_{t}^{1/2}\varrho_{i,t}(x)\leq  \sigma ,\forall x \in \mathcal{D}_{i}
\end{equation}
holds with finite number of sample time $T$.
\label{variancebound}
\end{lem}

\begin{proof}
According to Theorem 5 in \cite{srinivas2009gaussian}, for any bounded $\mathcal{X}_{i} \subset \mathcal{X}$, the information gain for exponential kernel is $\gamma_{T}=\mathcal{O}((\text{log}T)^{d+1})$, where $d$ is the dimension of $\mathcal{X}$. Lemma 7.1 in \cite{srinivas2009gaussian} shows that $\sum_{t=1}^{T}\min\{\sigma^{-2}\varrho_{i,t-1}^{2}(x),\alpha\}\leq \frac{2\alpha}{\text{log}(1+\alpha)}\gamma_{T}, \forall \alpha>0$. As the information gain $\gamma_{T}$ is convergent, such that the variance $\varrho_{i,t}(x)$ will be convergent to zero. As $\lambda_{t}$ is bounded, $ \lambda_{t}^{1/2}\varrho_{t}(x)$ will be convergent to zero, then  the lemma has been proved.
\end{proof}

Finally, we can prove that the algorithm presented in this paper can solve Problem \ref{mainproblem} with the relatively moderate Assumption \ref{assumption1}.

\begin{theorem}
Given a PWA system $\mathcal{S}$ with unknown dynamics, for any $\sigma>\varepsilon>0 $, the algorithm described in this paper (with the assumption that the standard deviation of the prediction error for the system identification component is bounded by $C$) can obtain an approximate abstract transition system $\hat{\mathcal{T}}$ that is $\sigma$-approximately simulated by the true abstract transition system $\mathcal{T}$ of $\mathcal{S}$ with a bounded probability that is greater than $1-\delta$, where $\delta=1-\frac{1}{\sqrt{2\pi(\sigma_{e}+C)}}\int^{\sigma-\varepsilon}_{-\sigma+\varepsilon}exp(-\upsilon^{2}/(2(\sigma_{e}+C)))d\upsilon$.
\label{abstractionerrorbound}
\end{theorem}

\begin{proof}
Set $\delta= \alpha_{2}$ and $\eta_{2}=\sigma-\varepsilon$. According to Lemma 3, Lemma 4, and Lemma 5, the system identification component together with the active sampling component can achieve an estimation $\parallel \hat{f}(x)-y\parallel\leq \eta_{2}$ with probability $> 1-\alpha_{2}$. Then with this estimation result, according to Lemma 1 and Lemma 2, the system abstraction component can generate an abstract transition system $\hat{\mathcal{T}}$ that is $(\eta_{2}+\varepsilon)-$approximately simulated by the real abstract transition system $\mathcal{T}$ with probability $> 1-\delta$.
\end{proof}

\section{Case Study}
\label{sec5}

In this section, we will use a soft robot driven by series pneumatic artificial muscles as an example to demonstrate our algorithm. 

\subsection{Model and LTL Specification}

\begin{figure}[hbtp]
\centering
\includegraphics[width=0.4\textwidth]{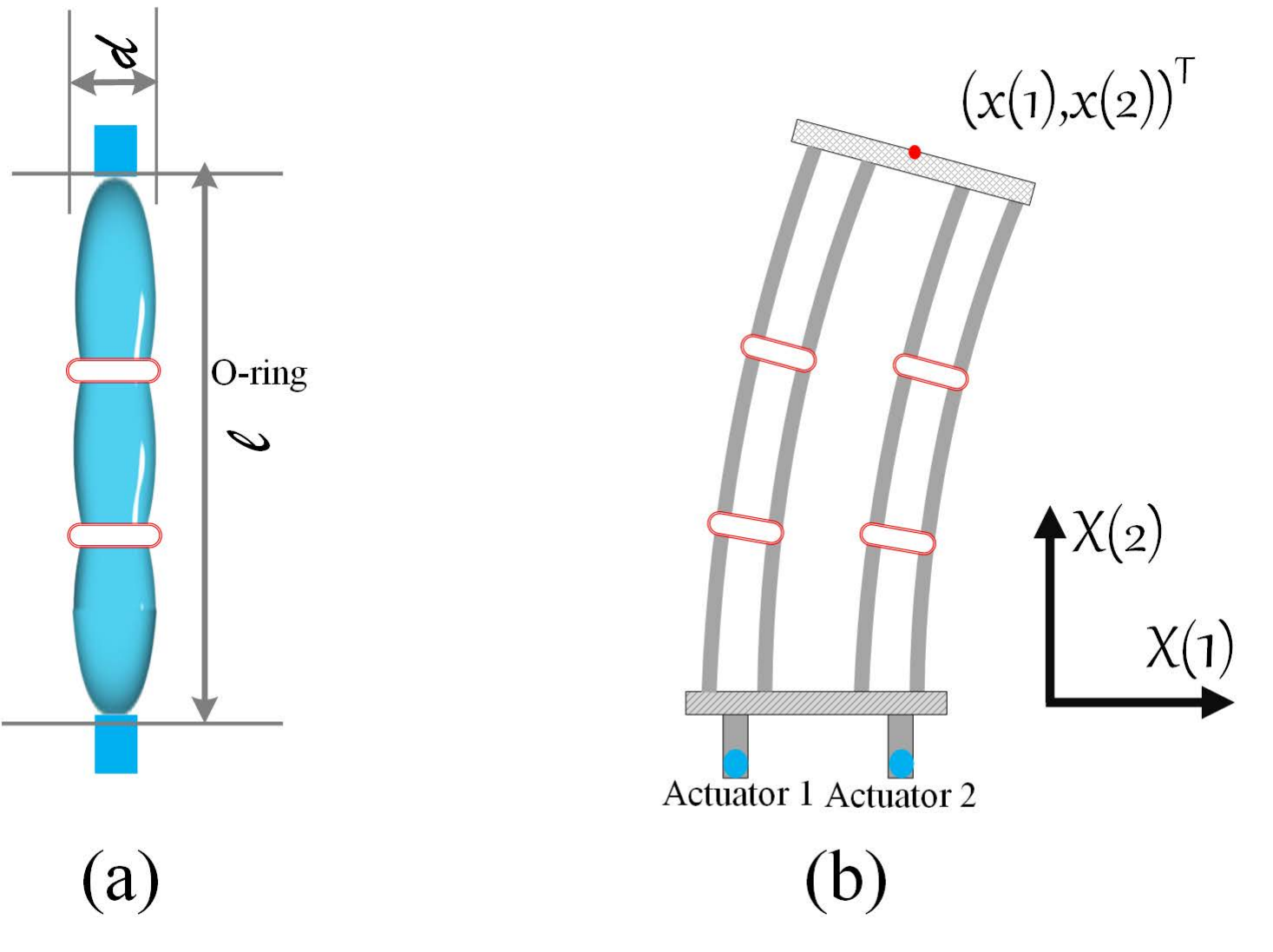}
\caption{(a) Geometrical model of the soft robot driven by series pneumatic artificial muscles (sPAM) used in the case study. (b) Simplified planar robot with $x=(x(1),x(2))^T$ as the coordinate vector of the moving platform.}
\label{Kinematicmodel}
\end{figure}

In this case study, we will focus on a particular type of soft robots, those that are driven by series pneumatic artificial muscles (sPAM) \cite{daerden2001concept}. Fig. \ref{Kinematicmodel} shows an example of such a robot. It has two polyethylene tubing sPAMs, each of which is controlled by a corresponding actuator with pressurized air. Even though the dynamics of such a robot is nonlinear, in  \cite{andrikopoulos2014piecewise}, the authors have shown that its closed-loop behavior can be approximated by piecewise affine dynamics. Here let's assume that the robot under investigation has dynamics as follows:
\begin{equation}
\label{PWApara}
x_{k+1} = f(x_k)+e
\end{equation}
and
\begin{equation*}
f(x)=\left\lbrace 
\begin{array}{lll}
\begin{bmatrix}
1 &0\\
0 & 0.98
\end{bmatrix}
x
\text{  if  } 0\leq x_{k}(1)\leq 0.3\\
\begin{bmatrix}
0.83 &0.12\\
0.12&0.81
\end{bmatrix}
x+ 
\begin{bmatrix}
0.01 \\
0.03
\end{bmatrix}
 \quad  \text{  if  }  x_{k}(1)\geq 0.3
\end{array}\right.
\end{equation*}
where $x=(x(1),x(2))^T$ is the coordinate of the moving platform and $e$ is a Gaussian noise with a zero mean and a standard deviation of 0.1. Please keep in mind that the model is unknown to our algorithm but the model itself can be used by our algorithm as a simulator to generate samples. 

The specification that needs to be verified is written as an LTL formula $\phi:=\square( \pi_{1} \wedge \diamondsuit\pi_{2})$, where  ``$\pi_{1}:=$ $x(1)$ is below 0.3", ``$\pi_{2}:=$ $x(2)$ is above 0.6", and $\square$ is the temporal operator ``Always'', $\diamondsuit$ is the temporal operator ``Eventually''. Put together, $\phi$ specifies that ``it should always be true that $x(1)$ is below 0.3 and eventually $x(2)$ is above 0.6''.

\subsection{Implementation Results}

The algorithm proposed in this paper is implemented as a Matlab tool. The tool takes an LTL formula and a black-box PWA system as inputs and it outputs  an abstract transition system. The active learning used in the tool is the Gaussian Process Adaptive
Confidence Bound (GP-ACB) algorithm proposed and implemented by our group in \cite{chen2016active}. Moreover, a Gaussian kernel function is used in GP-ACB. In order to demonstrate the effectiveness of our algorithm, here we will provide two sets of implementation results, one regarding the system identification component and the other one regarding the entire algorithm.

\subsubsection{Performance of the System Identification Component}

\begin{figure}
\begin{subfigure}{1\linewidth}
\centering
\includegraphics[width=0.9\textwidth]{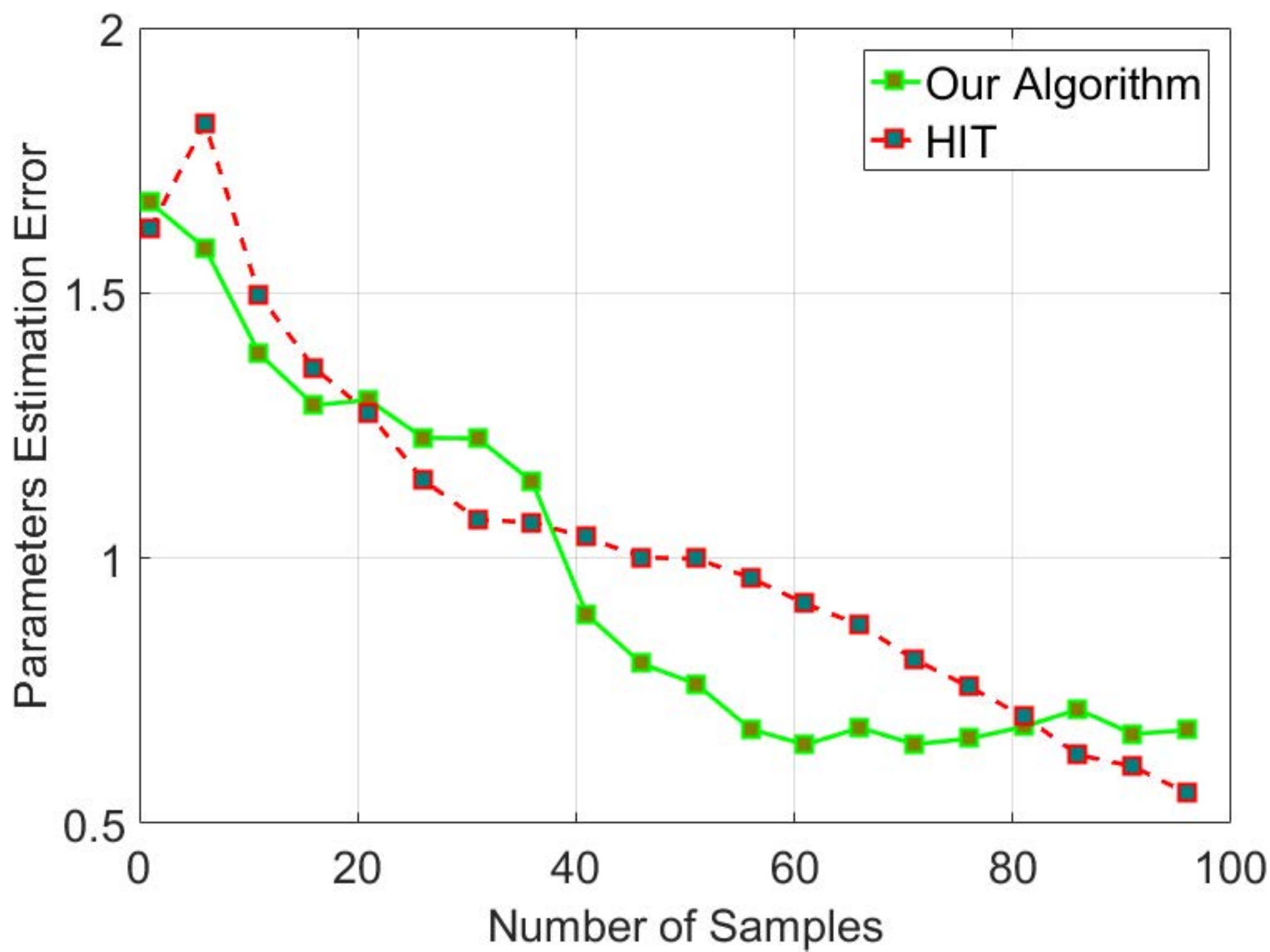}
\caption{}
\label{fig1:figure1}
\end{subfigure}
\begin{subfigure}{1\linewidth}
\centering
\includegraphics[width=0.9\textwidth]{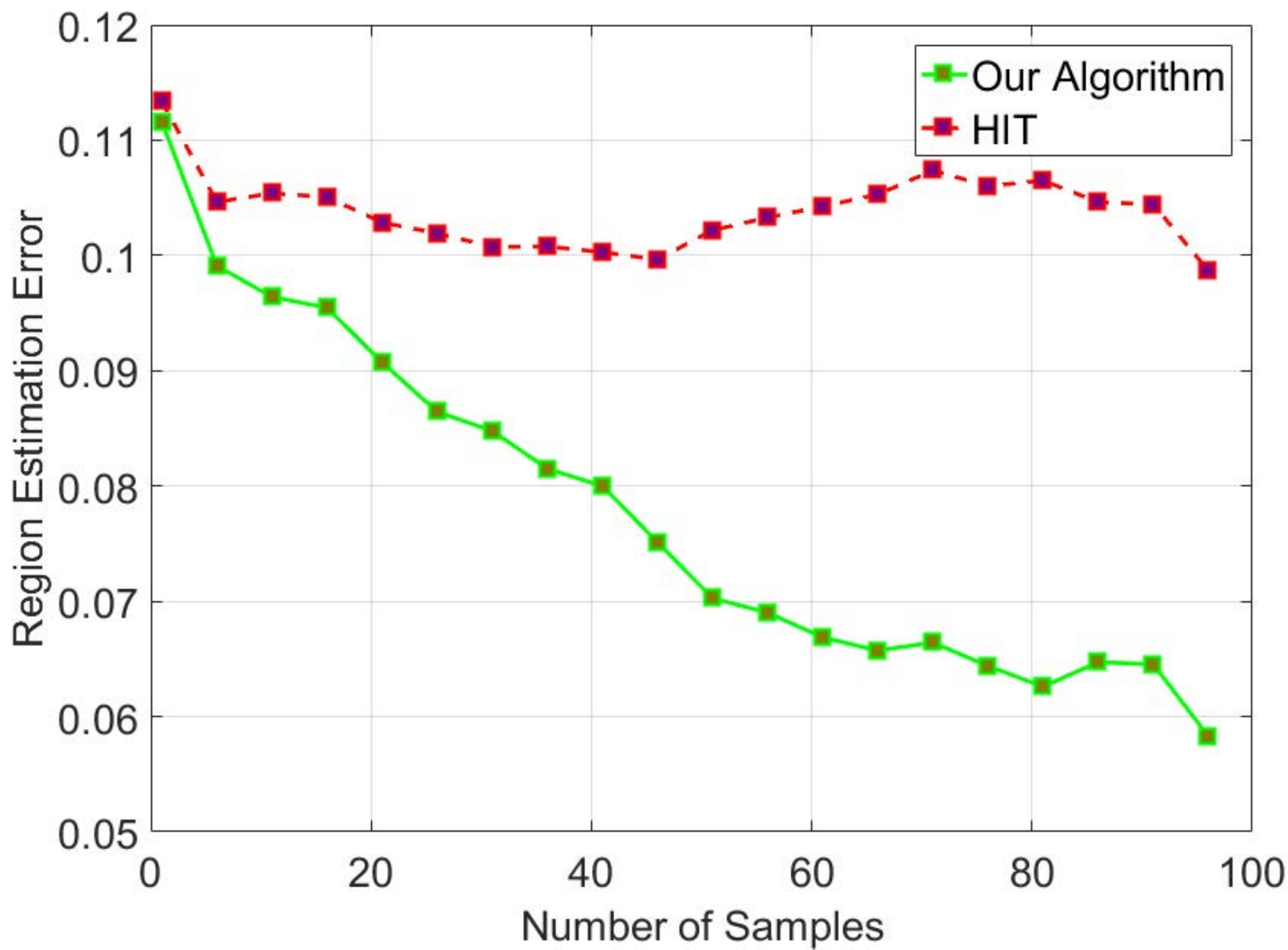}
\caption{}
\label{fig1:figure2}
\end{subfigure}
\caption{Comparison results of the system identification component used in this paper (green) and the HIT algorithm (red): (a) Average parameter estimation error with respect to the number of samples; (b) Average region (partition) estimation error with respect to the number of samples.}
\label{systemidentification}
\end{figure}

Here we compare the performance of our system identification component with an existing state-of-the-art black-box system identification tool called the Hybrid Identification Tool (HIT) \cite{HIT}. One thing we would like to point out here is that the evaluation of our system identification component is conducted in conjunction with other components, e.g., abstraction and active sampling, in the loop. It should be expected that, at worst, our component should have the same performance as the HIT. However, it may also be expected that factors such as active-learning based sampling and abstraction-guided refinement can potentially improve the identification performance, which turns out to be the case, at least for this particular case study. 

We use two metrics to quantify system identification errors: \textit{parameter estimation error} and \textit{region estimation error}. Given the parameters of a real PWA model (unknown to the investigated algorithms), the \textit{parameter estimation error} is the sum of the Euclidean distances between the real parameters and the estimated parameters. The \textit{region estimation error} is defined the sum of the following Hausdorff distance:
\begin{equation*}
e(\hat{\mathcal{X}}_{i},\mathcal{X}_{i})=\max\{\sup_{x\in \mathcal{X}_{i}} \inf_{y\in\hat{\mathcal{X}}_{i}} d(x,y),\sup_{y\in\hat{ \mathcal{X}}_{i}} \inf_{x\in\mathcal{X}_{i}} d(x,y)  \},
\end{equation*}
where $e(\hat{\mathcal{X}}_{i},\mathcal{X}_{i})$ the region estimation error related to the $i$th mode with $\mathcal{X}_{i}$ as the real region or partition and $\hat{\mathcal{X}}_{i}$ as the estimated one, $d(x,y)$ is the Euclidean distance between $x$ and $y$, and $\sup$ and $\inf$ stand for supremum and infimum, respectively. To calculate the Hausdorff distance, we randomly generate 100 samples inside the state space $\hat{\mathcal{X}}_{i}$ and the state space $\mathcal{X}_{i}$.

The comparison results based on 5 trials are shown in Fig. \ref{systemidentification}. It shows that, averagely speaking, the system identification component proposed in this paper has a comparable or better performance than HIT. Particularly, Fig. \ref{fig1:figure2} shows that our algorithm has a faster convergent rate regarding the region estimation error. This is probably due to the fact that, generally speaking, active learning, which is used in our algorithm, out-performs its randomly sampling counterpart, which is used in HIT. 

\subsubsection{Performance of the Entire Algorithm}

As the true abstract transition system of the system, described by Eqn. (\ref{PWApara}), is unknown, here we use the abstract transition system obtained by using the algorithm proposed in \cite{yordanov2013formal} as a benchmark. We will call this abstract transition system  as $\mathcal{T}^{*}$. In \cite{yordanov2013formal}, the authors have shown that even though their algorithm cannot attain the true abstract transition system, it still can get an abstract transition system that is arbitrarily close to the real one. Of course, we should point out that, in order to get this abstract transition system $\mathcal{T}^{*}$, the algorithm in \cite{yordanov2013formal} should have access to the system model, i.e., $\mathcal{T}^{*}$ is generated with a known model. In parallel, we run our algorithm to extract an abstract transition system $\hat{\mathcal{T}}$ without access to the dynamics of the model, i.e., $\hat{\mathcal{T}}$ is generated with a black-box system with unknown dynamics. Then in order to demonstrate the effectiveness of our algorithm, we need to show that $\hat{\mathcal{T}}$ should approach $\mathcal{T}^{*}$. This turns out to be the case, at least for this particular case study.  

Based on the system dynamics, Eqn. (\ref{PWApara}), and the LTL specification $\phi$, we implement the abstraction algorithm proposed in \cite{yordanov2013formal} and obtain an abstract transition system $\mathcal{T}^{*}$, which will be used as a benchmark. Then the algorithm proposed in this paper is applied to the same system, but with unknown dynamics, and the same LTL specification. The output of the algorithm is another abstract transition system $\hat{\mathcal{T}}$. Here we use a metric that is inspired by the concept of approximate simulation to quantify the difference between the two transition systems. To be more specific, we set the abstraction error to $\sigma$ if $\mathcal{T}^{*}$ is $\sigma-$approximately simulated by $\hat{\mathcal{T}}$. Moreover, the metric $\sigma$ is normalized to $\bar{\sigma}$ by the volume of the state space, i.e., $\bar{\sigma}:=\sigma/|\mathcal{X}|$.

\begin{table}[!htb]
\centering  
\caption{Comparison result with respect to different number of  samples and a fixed number of refinement steps, which is set to 20.}
\begin{tabular}{lccc }   
\hline
   &\multicolumn{3}{c}{Number of samples}\\    \hline 
           &20    &40 &60                        \\ \hline  %
 $\bar{\sigma} $&    0.046   &    0.042 & 0.035 \\    \hline
\end{tabular}
\label{results1}
\end{table}

\begin{table}[!htb]
\centering  
\caption{Comparison result with respect to different number of refinement steps and a fixed number of samples in active sampling component, which is set to 10.}
\begin{tabular}{lccc }   
\hline
   &\multicolumn{3}{c}{Refinement Steps}\\    \hline 
      &5   &10 &20       \\ \hline  %
 $\bar{\sigma}$ &0.077    &0.068  & 0.050   \\    \hline
\end{tabular}
\label{results2}
\end{table}

\begin{figure}[hbtp!]
\centering
\includegraphics[width=0.45\textwidth]{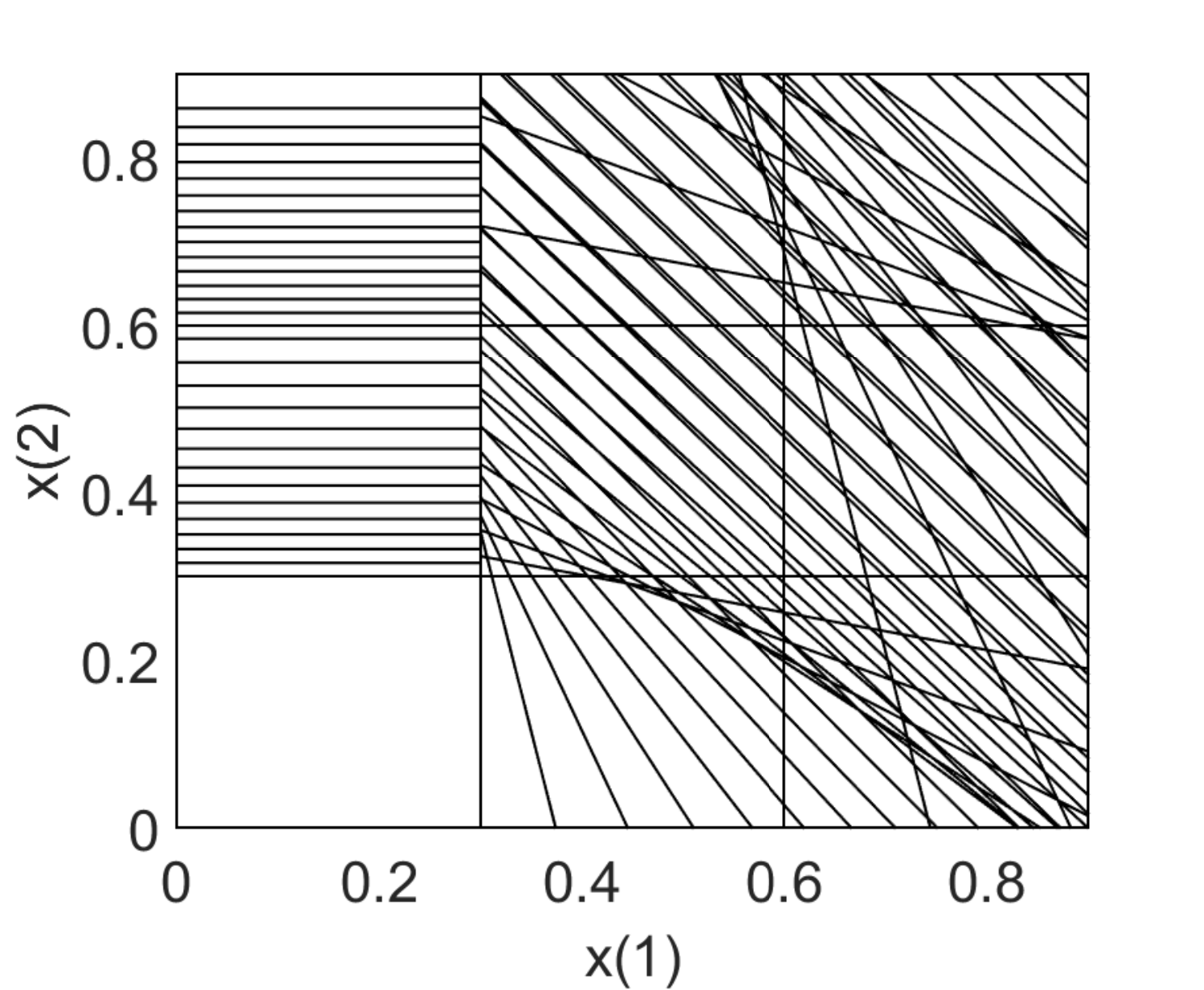}
\caption{State space partitions of the abstract transition system obtained by using our algorithm for the case study. The number of refinement steps and the number of samples in the active sampling component are both set at 20.}
\label{abstraction_results}
\end{figure}

The comparison results are shown in Table \ref{results1} and Table \ref{results2}. Table \ref{results1} shows the comparison results with respect to different number of samples and a fixed number of refinement steps in the system abstraction component. Table \ref{results2} shows the comparison results with respect to different number of refinement steps and a fixed number of samples in the active sampling component. The tables show the average normalized errors $\bar{\sigma}$ based on 5 trials. The results show that the larger the number of refinement steps, the smaller the abstraction error; and the larger the number of samples, the smaller the abstraction error. The results also show that, even with a black-box system, our algorithm can attain an approximate abstract transition system $\hat{\mathcal{T}}$ that is close to the benchmark abstract transition system $\mathcal{T}^{*}$.
 
\section{Conclusions}
\label{sec6}

In this paper, we proposed a data-driven approximate abstraction algorithm for piecewise affine systems with unknown dynamics. We demonstrated both theoretically and empirically that given a black-box PWA system and an LTL specification, we were able to derive an abstract transition system that is approximately simulated by the true abstract transition system. We demonstrated the effectiveness of our proposed algorithm with a soft robot as a case study. 

\section*{Acknowledgments}

This work was partially supported by the Hyundai Motor Company.

\bibliography{references}
\bibliographystyle{IEEEtran}
\end{document}